\documentclass[conference,letterpaper]{IEEEtran}

\usepackage{url}
\usepackage{ifthen}
\usepackage{cite}
\usepackage{amssymb}
\usepackage{amsmath}
\usepackage{amsfonts}
\usepackage{amsthm}
\usepackage{subcaption}

\usepackage{algorithmicx}
\usepackage{graphicx}
\usepackage{textcomp}
\usepackage{xcolor}
\usepackage{hyperref}
\usepackage{algorithm}
\usepackage{algpseudocode}
\usepackage[font=small,labelfont=bf]{caption}

\newcommand{\rev}[1]{\textcolor{black}{#1}}
\newcommand{\nrev}[1]{\textcolor{black}{#1}}

\newcommand{\supply}[1]{\textcolor{black}{#1}}

\newtheorem{theorem}{Theorem}
\newtheorem{corollary}{Corollary}
\newtheorem{lemma}{Lemma}

\newcommand{\E}[1]{\mathbb{E}[#1]}

\begin{document}
\title{Improved Interactive Protocol for Synchronizing From Deletions
\vspace{-0.5em}} 

% %%% Single author, or several authors with same affiliation:
% \author{%
%  \IEEEauthorblockN{Author 1 and Author 2}
% \IEEEauthorblockA{Department of Statistics and Data Science\\
%                    University 1\\
 %                   City 1\\
  %                  Email: author1@university1.edu}% }

%%% Several authors with up to three affiliations:
% \author{%
%   \IEEEauthorblockN{Haolun (Michael) Ni, Lev Tauz and Lara Dolecek}
%   \IEEEauthorblockA{Department of Electrical and Computer Engineering \\
%                     University of California, Los Angeles\\
%                     Los Angeles\\
%                     Email: %missing%
%                     }
%   \and
%   \IEEEauthorblockN{Ryan Gabrys}
%   \IEEEauthorblockA{California Institute for Telecommunication \\and Information Technology\\ 
%   University of California, San Diego
%                     San Diego\\
%                     Email: %missing%
%                     }
% }

\author{\IEEEauthorblockN{{Haolun (Michael) Ni\IEEEauthorrefmark{1}, Lev Tauz\IEEEauthorrefmark{1},
Ryan Gabrys \IEEEauthorrefmark{2},
Lara Dolecek\IEEEauthorrefmark{1}} \\
\IEEEauthorblockA{\small{\IEEEauthorrefmark{1} Department of Electrical and Computer Engineering, University of California, Los Angeles, USA}}
\IEEEauthorblockA{ \IEEEauthorrefmark{2}
        \small{ California Institute for Telecommunication and Information Technology, San Diego, USA}}
 \IEEEauthorblockA{
		\small{michaelni12@g.ucla.edu, levtauz@ucla.edu, rgabrys@ucsd.edu}, \small{dolecek@ee.ucla.edu}}
		\vspace{-3em}}

}

\maketitle

\begin{abstract}
Data synchronization is a fundamental problem with applications in diverse fields such as cloud storage, genomics, and distributed systems. This paper addresses the challenge of synchronizing two files, one of which is a subsequence of the other and related through a constant rate of deletions, using an improved communication protocol. Building upon prior work, we integrate advanced multi-deletion correction codes into an existing baseline protocol, which previously relied on single-deletion correction. Our proposed protocol reduces communication cost by leveraging more general partitioning techniques as well as multi-deletion error correction. We derive a generalized upper bound on the expected number of transmitted bits, applicable to a broad class of deletion correction codes. Experimental results demonstrate that our approach outperforms the baseline in communication cost. These findings establish the efficacy of the improved protocol in achieving low-redundancy synchronization in scenarios where deletion errors occur.
\end{abstract}

\section{Introduction}
\label{section:introduction}

% Data synchronization is a critical research topic with applications in diverse domains such as cloud data storage and DNA analysis. At its core, a data synchronization problem involves two versions of a file: the "correct" version and a "corrupted" version. The corrupted version differs from the correct version due to certain edits, which may include deletions, insertions, or substitutions. The goal of data synchronization is to accurately identify and correct these edits through communication between the two file versions. Depending on the type of edits to be addressed, different synchronization protocols have been developed. \textcolor{red}{I wouldn't phrase it this way. There doesn't have to be a ``correct'' version and manytimes in bio applications or even when dealing with files you just want to know the deltas between the two files. Just say there's two files and you want to know exactly how to convert one to the other (which is the same as knowing the deltas)}

Data synchronization is a critical research topic with applications in diverse domains such as cloud data storage and DNA analysis. At its core, a data synchronization problem involves two versions of a file that may differ due to edits such as deletions, insertions, or substitutions. The objective is to precisely determine how to transform one file into the other by identifying the differences between them. Depending on the nature of the edits, different synchronization protocols have been developed towards achieving this goal efficiently and accurately \rev{\cite{or_code, VT, haeuplar_code}}. This work focuses on the scenario where the edits are deletions. Due to the symmetric nature between insertions and deletions, and that substitutions can be viewed as a combination of a deletion and an insertion, developing a deletion correction protocol serves as a fundamental step toward creating protocols capable of handling a broader range of edits.

Since the introduction of the deletion correction problem by Levenshtein in the 1960s, various approaches have been proposed to improve synchronization using advanced mathematical techniques. Levenshtein demonstrated that the Varshamov-Tenengolts (VT) codes \cite{ VTalone} could correct a single deletion or insertion \nrev{\cite{Levenshtein_code}}. A one-way, error-free coding scheme correcting up to \( \delta \) deletions and insertions \nrev{for a length $n$ sequence} with \( O(\delta \log n) \) redundancy was introduced in \cite{or_code}, albeit with high decoding complexity. A low-complexity protocol achieving similar redundancy was presented in \cite{VT}. An interactive, deterministic protocol was developed in \cite{ori} based on the work in \cite{VT}, extending its capability to correct a fixed rate of deletions. This protocol was later enhanced in \cite{ori_ex} to address both deletions and insertions, as well as to support non-binary sequences. \rev{The problem was further extended to approximate synchronization in \cite{itw_paper}, where instead of requiring the reconstructed file to be exactly the same as the original one, a predefined rate of distortion was accepted. While the protocols in \cite{ori, ori_ex, itw_paper} focus on the probabilistic model, \cite{haeuplar_code} introduced a single-round, optimal coding scheme for the worst-case model, where the number of deletions is fixed rather than probabilistic.} Other recent advances in data synchronization include the one-way multilayer code \cite{multilayer}, polynomial-time decodable codes \cite{polycode}, localized deletion correction codes \cite{localcode}, and single-round synchronization schemes \cite{singleround}.

The interactive protocol introduced in \cite{ori} provides a practical approach for handling a \textit{fixed rate} of deletions \nrev{(rather than a \textit{fixed number} of deletions)}, making it capable of correcting a large number of deletions as long as their count remains linear with the sequence length. The protocol employs the VT codes for deletion correction, which requires iteratively dividing the sequence into smaller parts, each containing at most one deletion. However, since the publication of the baseline protocol in \cite{ori}, advances in deletion correction codes have introduced schemes capable of correcting multiple deletions with reasonable redundancy for a specified number of deletions. Examples of such multi-deletion correction codes include the coding schemes presented in \cite{2d,ryan_2d_code,ryan_general_code}. These advancements highlight the potential for further improving the baseline protocol by integrating multi-deletion correction codes.

The primary objective of this work is to explore the integration of multi-deletion correction codes into the baseline interactive protocol to develop an improved protocol and evaluate its performance relative to the baseline protocol across key communication metrics, including the number of bits transmitted and computational complexity. The main contributions of this work are summarized as follows:

\begin{itemize}
    \item \textbf{Incorporation of Multi-Deletion Correction Codes:} \rev{Multi-deletion correction codes are leveraged alongside the Varshamov-Tenengolts (VT) codes to address deletions in the sequence, mitigating the extensive repetition required when dividing the sequence. This enhancement reduces the number of bits transmitted while maintaining the computational complexity of the improved protocol within the same order of magnitude as the baseline protocol, provided that the computational complexity of the employed multi-deletion correction codes does not exceed this level.}

    \item \textbf{Strategic Adjustment of Segment Length:} \rev{\nrev{The protocol repeatedly divides the file to be synchronized into smaller portions.} The length of each portion after the initial division is strategically adjusted to further minimize the number of transmitted bits. This \nrev{adjustment} introduces a new tunable parameter in the protocol, which allows for optimization based on the specific characteristics of the deletion correction codes employed and the deletion rate \( \beta \). By carefully configuring this parameter, the synchronization process achieves even lower communication costs.}

    \item \textbf{Refinement of Estimation Techniques:} \rev{Certain estimations used to derive the upper bound \nrev{on redundancy} in the baseline protocol were identified as suboptimal. Advanced estimation techniques are employed to refine and tighten this upper bound, contributing to the protocol's sharper communication cost estimates.}
\end{itemize}

As a result of these combined contributions, the improved protocol achieves a significant reduction in the theoretical upper bound on the number of bits transmitted. For a file sequence of length \( n \) and a deletion rate \( \beta \), the baseline protocol transmits at most \( 109n\beta \log \frac{1}{\beta} \) bits \nrev{as stated in \cite{ori}}, while the improved protocol reduces this upper bound to \( 28.5n\beta \log \frac{1}{\beta} \).

To illustrate the practical benefits of these enhancements, we also implement the improved protocol using the VT codes for single-deletion correction and the two-deletion correction code introduced in \cite{2d}. Experimental results demonstrate that the improved protocol outperforms the baseline protocol in the number of bits transmitted. This finding highlight the effectiveness of integrating multi-deletion correction codes in reducing the communication cost of data synchronization. 
% \textcolor{red}{Comment(resolved? see the blue line in the second paragraph):\url{https://arxiv.org/abs/1804.03604}, Deterministic Document Exchange Protocols, and Almost Optimal Binary Codes for
% Edit Errors, There are optimal schemes that take one round. Different models. These papers give optimal schemes but it's for a worst case model, not probabilitic }

The remainder of the paper is organized as follows: \rev{Section \ref{section:background} covers relevant background information.} Section \ref{section:algorithm} details the protocol and highlights key differences from the baseline. Section \ref{section:theory} analyzes the protocol theoretically, including its upper bound on transmitted bits and complexity. Section \ref{section:experiment} presents experimental results comparing the improved and baseline protocols.

\section{Background}
\label{section:background}

Consider two parties: Alice, the encoder, who possesses the original file \( X \), a binary sequence of length \( n \), and Bob, the decoder, who holds a version \( Y \) derived from \( X \) with certain bits randomly deleted. In this work, each bit in \( X \) is modeled as an independent and identically distributed random variable, with each bit having a probability \( \beta \) $(\beta>0)$ of being deleted. Communication between Alice and Bob occurs over a noiseless, two-way channel, enabling multiple rounds of interaction until synchronization is achieved. The primary objective of the protocol is to minimize the number of bits transmitted. The complexity of the protocol is another key criterion, as it determines the computational resources needed for encoding, decoding, and overall execution—factors critical to practical implementations. Accordingly, this work evaluates the performance of the improved protocol using the number of bits transmitted and the computational complexity.

Before delving into the details of the improved synchronization protocol, it is worthwhile to briefly revisit the baseline protocol introduced in \cite{ori}, as the improved protocol builds upon it. The baseline protocol addresses deletion correction through three modules. In the first module, the sequence is divided into smaller portions, enabling a more manageable expected number of deletions within each portion compared to the entire sequence. In the second module, deletions within each portion are corrected individually. Since the baseline protocol utilizes VT codes, which can only correct one deletion at a time, an additional step is required to further divide portions where the actual number of deletions exceeds this capability. This iterative process of dividing portions into smaller subsequences until VT codes can be applied is referred to as the divide-and-conquer approach, as introduced in \cite{VT}. Following the correction of deletions using the divide-and-conquer approach, a third module is implemented to handle residual errors from the first two modules, \nrev{which can at this point be viewed as substitution errors, thus} ensuring the protocol achieves error-free synchronization.

\section{Improved Synchronization Algorithm}
\label{section:algorithm}

Similar to the baseline protocol, the improved protocol consists of three core modules: the Matching Module, the Deletion Recovery Module, and the Error Correction Module. Together, these modules synchronize Bob’s sequence to match Alice’s, ensuring both efficiency and accuracy. The improved protocol introduces three key advancements: \rev{(1) multi-deletion correction codes are integrated with VT codes to reduce sequence repetition, lowering the number of transmitted bits without increasing the overall computational complexity; (2) segment lengths are strategically adjusted to optimize communication costs based on deletion rates and the properties of the employed correction codes; and (3) refined estimation techniques are used to tighten the upper bound on communication cost, further enhancing protocol efficiency.} The following subsections provide a detailed explanation of each module.

\subsection{The Matching Module}

The Matching Module partitions the sequence into smaller portions to reduce the number of deletions in each portion. To achieve this goal, Alice transmits specific parts of her sequence, referred to as \emph{pivots}, to Bob. The bits between two consecutive pivots are called \emph{segments}. The lengths of both pivots and segments are predetermined, with the pivot lengths being significantly shorter than the segments. In this work, the segment length \( L_S \) is defined as \( L_S = \frac{s}{\beta} \), where \( \beta \) is the deletion rate, and \( s \) is a tunable parameter called segment length multiplier. In the baseline protocol, \( s \) was empirically set to \( 1 \). However, we observe that modifying the segment length could further reduce the number of bits transmitted. Consequently, \( s \) is treated as an adjustable parameter in this work to enable additional optimization. The pivot length \( L_P \) is chosen as the smallest integer satisfying 
\( L_P \geq 3s + 8 + 2 \log \frac{1}{\beta} \). This selection is designed to minimize the probability of errors, and a detailed explanation of this restriction is provided \supply{in the Appendix.} Once the sequence is divided into pivots and segments, Alice transmits the entire set of pivots to Bob through the communication channel.

Upon receiving the pivots, Bob attempts to locate them in his sequence. Note that Bob may not be able to find all the pivots; some may be missing due to deletions. Additionally, errors may arise if Bob misidentifies a pivot because of the presence of a \emph{false pivot}, which refers to a coincidental match elsewhere in the sequence. To address this issue, Bob performs a graph-based optimization to maximize the selection of correct pivots while minimizing false matches. This graphing process, detailed in \cite{ori}, remains unchanged in the improved protocol. By ensuring that the pivot length satisfies the requirement mentioned above, matching errors are effectively constrained.

After completing the optimization process, Bob informs Alice of the selected pivots, enabling both parties to partition their sequences into corresponding \emph{sections} for subsequent processing. \nrev{This step} effectively divides the original sequence into smaller sections, significantly reducing the expected number of deletions per section compared to the entire sequence. As a result, the computational complexity of the Deletion Recovery Module is substantially lowered. Additionally, since all sections can be processed in parallel, the protocol's latency can be significantly reduced.

\subsection{The Deletion Recovery Module}

Following the Matching Module, the sequences are partitioned into smaller sections with fewer deletions. The Deletion Recovery Module focuses on correcting deletions within these sections. For the improved protocol, assume a set of deletion correction codes capable of handling up to \( w \)-deletions is utilized, where \( w \) is a positive integer. \nrev{This assumption indicates} that the protocol can directly correct up to \( w \)-deletions in a section without extra procedures. However, it is possible for sections to have a number of deletions exceeding this limit. To address these scenarios, the divide-and-conquer approach is employed.

As previously mentioned, this approach recursively splits sections into smaller parts until the deletions in each part fall within the correction capability of the codes. The decoder computes the number of deletions in each section and communicates with the encoder accordingly. If deletions exceed $w$, the encoder transmits \emph{delimiters} — short strings of bits around the section center — to aid in further splitting the sequence. The delimiter length is set to $c \log L_s$, where $c$ is the delimiter length coefficient typically set to 3 unless otherwise stated, \nrev{and $L_s$ is the length of the section}. If a delimiter is not found due to deletions, the decoder requests another delimiter until one is identified. Once a delimiter is found, the current part of the sequence is divided into two smaller halves, each with the number of deletions no more than the original part. If the number of deletions in a half is no more than $w$, the corresponding deletion correction code is applied. In this scenario, suppose there are \( i \) deletions in a half (\( i \leq w \)). The protocol employs the deletion correction code designed to handle these \( i \) deletions. Otherwise, the half is further divided. This process continues until all parts of the section are synchronized. Detailed explanation of the divide-and-conquer approach could be found in \cite{VT}.

The baseline protocol follows the same procedure for correcting deletions in the Deletion Recovery Module, with the distinction that it employs the VT code, which can only correct a single deletion at a time. Thus, in the baseline protocol, the factor $w$ is set to 1, limiting its deletion correction capability. In contrast, the introduction of multi-deletion correction codes in the improved protocol reduces the frequency of partitioning and minimizes the number of delimiters that need to be transmitted. This results in fewer rounds and less bits transmitted as delimiters. Detailed discussion of how the incorporation of the multi-deletion correction codes impacts the communication cost can be found in Section \ref{section:theory} and \ref{section:experiment}.

Errors in this module can occur if a delimiter is incorrectly matched, leading to improper partitioning. However, the probability of such an error is upper-bounded. When the delimiter length coefficient $c$ is set to 3, an 
appropriate upper bound for the error probability in the deletion recovery module is provided in the Appendix.

\subsection{The Error Correction Module}

After the first two modules, the deletions are corrected, yielding a preliminarily synchronized sequence. \textcolor{black}{However, the sequence obtained at this stage is not guaranteed to be an exact copy of the encoder’s sequence.} As previously stated, errors may occur in the previous modules due to mismatched pivots or delimiters, and these errors manifest themselves as substitutions. Residual errors from these modules necessitate a final error correction step. \supply{In the Appendix}, it is proven that after the first two modules, the probability of error in the resulting sequence is upper-bounded by $2\beta + o(\beta)$. Given this, the error can be corrected by Alice sending  approximately \( nH(2\beta + o(\beta)) \) (i.e. the capacity of the substitution channel) bits to Bob as ordinary parity check bits, where \( H(\cdot) \) denotes the binary entropy function. Once the residual errors are corrected in the Error Correction Module, the sequence Bob receives will match Alice's sequence exactly. Thus, the entire deletion correction process is completed.

\section{Analysis and Discussion} \label{section:theory}

This section presents the theoretical analysis of the improved protocol, with a primary focus on establishing an upper bound for the total number of bits transmitted. This upper bound is compared with that of the baseline protocol to highlight the improvements introduced by the enhanced design. Additionally, the computational complexity of the protocol is analyzed to demonstrate its feasibility for practical implementation and application.

% \subsection{Upper Bound of Bits Transmitted}
The following theorem provides an upper bound on the number of bits transmitted in the improved protocol:

\begin{theorem} \label{theorem:improved}
    Let \( s > 0 \) be the segment length multiplier. Let \nrev{\( c>0 \)} be the delimiter length coefficient in the Deletion Recovery Module. Let \( w \) be the maximum number of deletions that we can correct with a set of multi-deletion correction codes. Let \( a_i \geq 1, i \in [1, \dots, w] \), be the efficiency of the \( i \)-deletion correction code, where the \( i \)-deletion correction code requires \( i a_i \log q \) bits to correct \( i \) deletions in a sequence of length \( q \). Define \( a = \max\{a_1, a_2, \ldots, a_w\} \). For our improved synchronization protocol, the average number of bits transmitted is upper-bounded by  
    \[
    2\frac{s+1}{s} \left(\frac{2^w}{2^w-1} c + a + 2 \right) n \beta \log \frac{1}{\beta}.
    \]  
\end{theorem}

% \todo{Comment(resolved, see below):Need to describe and justify the use of $a_i$ in the protocol. Specifically that it is for accounting for the suboptimality of practical codes in the multi-deletion code literature. Reference the paper for the lower bound.}

\rev{The efficiency \( a_i \) of the \( i \)-deletion correction code is introduced in Theorem \ref{theorem:improved} to account for the suboptimality of practical deletion correction codes in terms of redundancy. For an \( i \)-deletion correction code, the lower bound on redundancy indicates that at least \(i\log q\) bits are required to be transmitted to correct \(i\) deletions for a sequence of length \( q \). While this \nrev{expression} represents a theoretical lower bound, existing multi-deletion correction codes for \( i > 1 \) have not yet achieved this level of efficiency. The parameter \( a_i \geq 1 \) is therefore introduced as an efficiency coefficient to quantify how closely a given \( i \)-deletion correction code approaches this bound. Specifically, \( a_i = 1 \) indicates optimal redundancy, and higher values reflect greater deviation from the theoretical lower bound. Importantly, \( a_i \) can be adjusted as new, more efficient multi-deletion correction codes are developed.}

The upper bound for Theorem \ref{theorem:improved} is derived by upper bounding the number of bits transmitted in each module of the protocol. Through stochastic analysis, we can sufficiently bound both the number of bits transmitted and the probability of incorrect decoding in the first two modules. \nrev{Hence the number of bits needed in the error correction module to fix any potential decoding mistakes can also be bounded.} Detailed proofs can be found \supply{in the Appendix}. The remainder of this section will focus on discussing the improvements demonstrated by Theorem \ref{theorem:improved}.

For comparison, we present the upper bound of the baseline protocol as proved in \cite{ori}.

\begin{theorem}[from Lemma 1, \cite{ori}]
\label{theorem:baseline}
    Let $c$ be the delimiter length coefficient in the Deletion Recovery Module. For the baseline  synchronization protocol where $s=1, w=1$ and $a=1$, the average number of bits transmitted is upper-bounded by  
\[
\left(8(4c+1)+5\right) n \beta \log \frac{1}{\beta}.
\]  
\end{theorem}

% \todo{Comment(resolved): Add comparison with theorem 2 (s =1, w =1 ,a = 1) and show that when the protocol is the baseline protocol, we have a tighter bound. Thus, we will compare the improvements of generalization on our bound.}

\rev{First, we want to highlight the refined upper bound of Theorem \ref{theorem:improved} in comparison with Theorem \ref{theorem:baseline}. The baseline protocol, which employs only the VT codes as the deletion correction code, corresponds to the case where \( w = 1 \), \( s = 1 \), \( c = 3 \), and \( a = 1 \). Substituting these values into the bound in Theorem \ref{theorem:improved}, the upper bound for the baseline protocol simplifies to \( 36n \beta \log \frac{1}{\beta} \). In contrast, the upper bound from Theorem \ref{theorem:baseline} for \( c = 3 \) is \( 109n \beta \log \frac{1}{\beta} \), as shown previously. This \nrev{result} demonstrates that even without introducing multi-deletion correction codes or modifying the segment length, the new upper bound tightens the result on the baseline protocol.}

\rev{Next, we shall discuss the benefits of utilizing multi-deletion correction codes and adjusting segment length in the Deletion Recovery Module. Consider the following practical scenario where \nrev{the multi-deletion correction codes can only correct up to two deletions (i.e., $w=2)$}. To the best of the authors' knowledge, the two-deletion correction code in \cite{2d} is among the most redundancy-efficient two-deletion correcting codes with a redundancy of \( 7 \log q \) for a sequence of length \( q \). For the single deletion correcting code, we can utilize the redundancy optimal VT codes. With the parameters set as \( s=2 \), \( w=2 \), \( c=3 \), and \( a=3.5 \), the upper bound on the number of bits is  \( 28.5n\beta \log \frac{1}{\beta} \) which is a significant improvement over the baseline protocol upper bound of \( 36n \beta \log \frac{1}{\beta} \), even with a suboptimal two-deletion correcting code.  } 
% Another noteworthy numerical result of the upper bound in Theorem \ref{theorem:improved} involves the case where the VT code and the two-deletion correction code from \cite{2d} are employed as deletion correction codes. This specific configuration is also implemented in the next section. To the best of the author's knowledge, the two-deletion correction code in \cite{2d} is among the most advanced multi-deletion correction codes, achieving a low redundancy of \( 7 \log q \) to correct two deletions in a sequence of length \( q \). Setting the segment length multiplier \( s \) to 2, and substituting \( s=2 \), \( w=2 \), \( c=3 \), and \( a=3.5 \) into the derived formula, the upper bound for the implemented improved protocol becomes \( 28.5n\beta \log \frac{1}{\beta} \), consistent with the result stated in Section \ref{section:introduction}. It is worth noting that the redundancy of the two-deletion correction code is not yet optimal, indicating that the upper bound could be further improved with the development of more advanced deletion correction codes.

\nrev{To evaluate the efficiency of the improved protocol, we define the redundancy coefficient \( r(s, w, a, c) \) as:}
\begin{equation}
    r(s, w, a, c) = 2\frac{s+1}{s} \left(\frac{2^w}{2^w-1} c + a + 2 \right),
\end{equation}
\nrev{This coefficient captures the leading term in the upper bound of the communication cost, specifically the factor in front of \( n\beta\log \frac{1}{\beta} \), as established in Theorem \ref{theorem:improved}.}

\nrev{The redundancy coefficient \( r(s, w, a, c) \) is independent of the sequence length \( n \) and deletion rate \( \beta \), making it a versatile metric for comparing protocol performance under varying configurations. Figure \ref{r_s} illustrates the relationship between \( r(s, w, a, c) \), the segment length multiplier \( s \), and the maximum deletions \( w \) the codes can correct. Notably, increasing \( s \) reduces the theoretical upper bound on redundancy. However, as experimental results will show, actual redundancy does not always decrease monotonically; after a certain point, larger segment lengths can increase the redundancy of the protocol. Additionally, Fig. \ref{r_s} highlights the impact of using more robust deletion correction codes. For a fixed \( a \), increasing \( w \) leads to a lower redundancy coefficient, aligning with expectations since stronger codes reduce the overall communication overhead. These insights underline the importance of selecting \( s \) and \( w \) to enhance theoretical efficiency and practical performance.}

\begin{figure}[t]
\centerline{\includegraphics[width=0.8\linewidth]{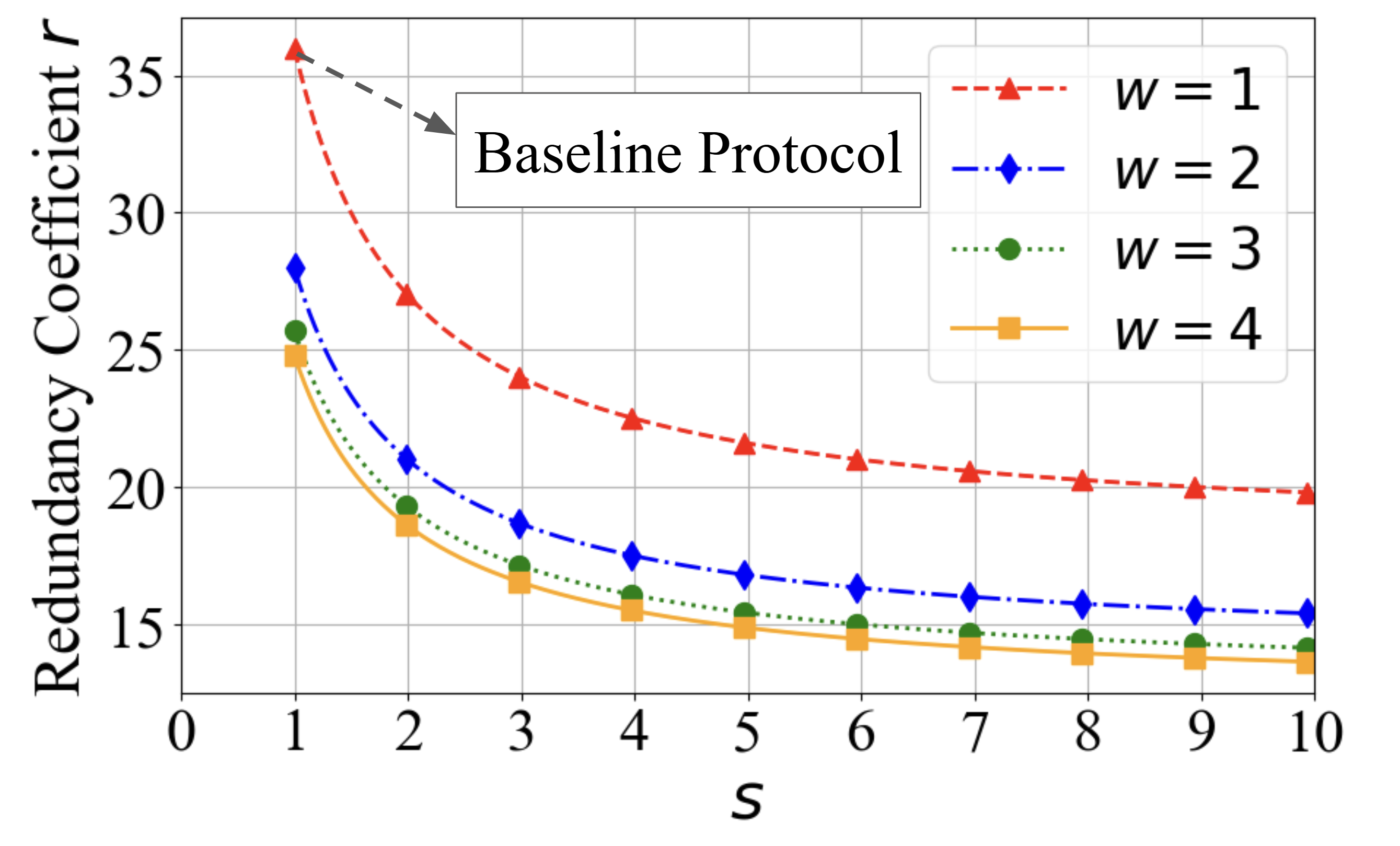}} % Adjusts the width to fit within the text width
\caption{Graph of the redundancy coefficient \( r \) as a function of \( s \) for various values of \( w \) with fixed parameters \( a = 1 \) and \( c = 3 \).}
\vspace{-0.4cm}
\label{r_s}
\end{figure}

The computational complexity of the improved protocol primarily depends on the deletion correction codes employed. Excluding the deletion correction process, the most demanding step is the pivot selection in the Matching Module, with a theoretical complexity of \( O(n^4\beta^6) \), as estimated in \cite{ori}. This matches the baseline protocol's complexity. For the implemented protocol, which uses VT codes and a two-deletion correction code with \( O(n) \) complexity, the overall complexity remains dominated by the Matching Module, resulting in \( O(n^4\beta^6) \). While increasing the segment length may reduce practical computational costs, the theoretical complexity remains unchanged.

\section{Experimental Results}
\label{section:experiment}

This section presents the experimental results for the improved protocol. As mentioned earlier, the implemented protocol uses VT codes and two-deletion correction codes to address single and double deletions in a string, respectively. The performance of the improved protocol is compared with the baseline protocol \rev{which employs VT codes for correcting single deletions.} \rev{Note that the simulation results utilize the channel capacity for the number of bits transmitted in the Error Correction Module. This focus is justified by the observation that the majority of the communication cost arises from the first two modules. Consequently, our simulations and analysis primarily emphasize these two modules, as they represent the key contributors to the overall communication cost.} Since the segment length plays a crucial role in determining the number of bits transmitted in the protocol, the results for different segment lengths will be presented and compared for both protocols first, with an analysis of how varying segment lengths impact performance.

\begin{figure}[t]
    \centering
    % First subfigure
    \begin{subfigure}[t]{0.48\linewidth} % Adjust width as needed
        \includegraphics[width=\linewidth]{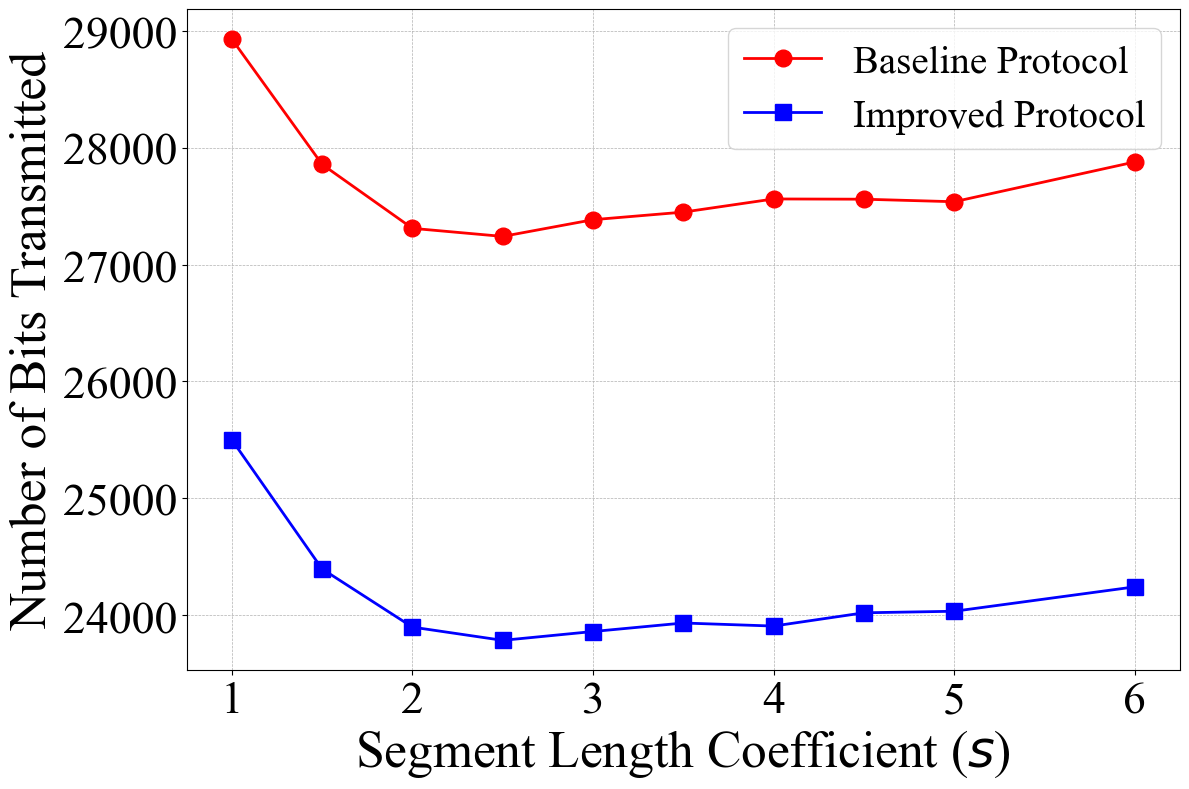}
        \caption{\small Deletion rate $\beta = 0.01$.}
        \label{fig5}
    \end{subfigure}
    \hfill
    % Second subfigure
    \begin{subfigure}[t]{0.48\linewidth} % Adjust width as needed
        \includegraphics[width=\linewidth]{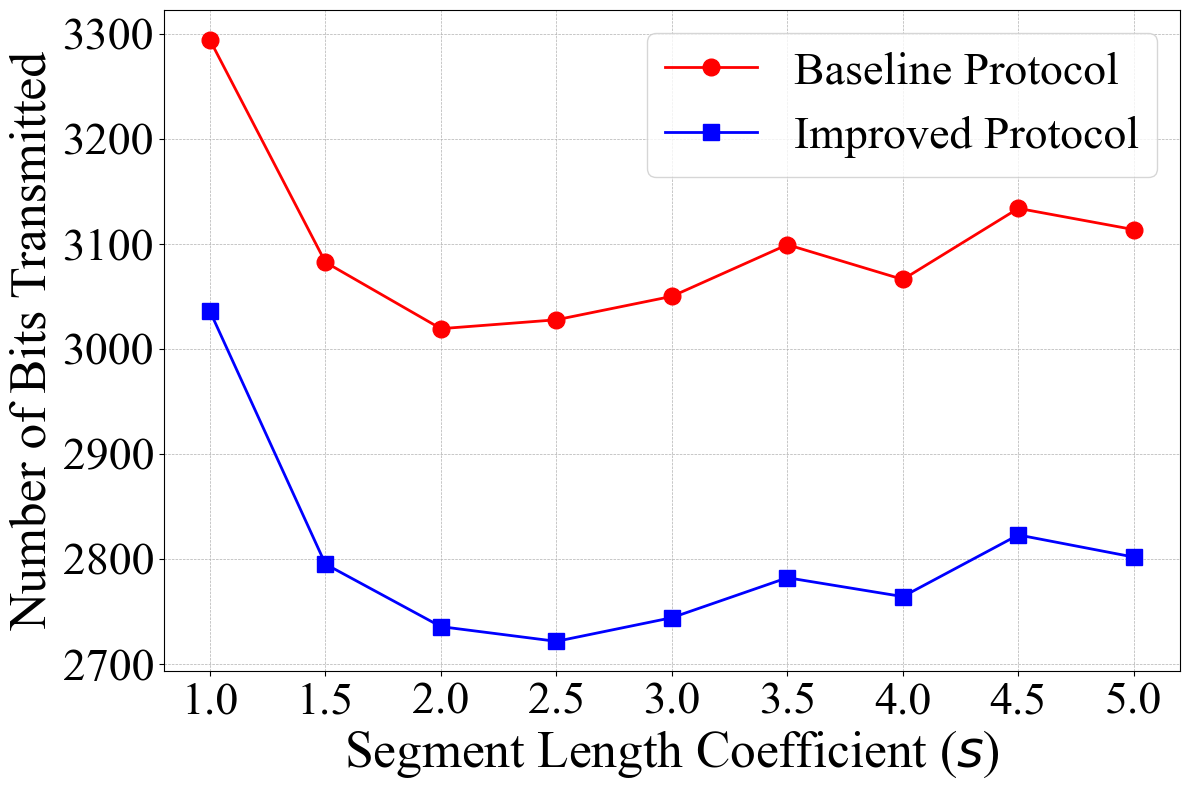}
        \caption{\small Deletion rate $\beta = 0.001$.}
        \label{fig6}
    \end{subfigure}
    \caption{\small Number of bits transmitted for both the baseline and improved protocols across varying segment length coefficients, where the segment length is scaled by the coefficient. Each subfigure represents a different deletion rate.}
    \vspace{-0.4cm}
    \label{fig:combined}
\end{figure}

% \todo{Comment(resolved?, see above):Explicitly state which are the improved and baseline protocol parameters}

% Furthermore, the major technical contributions of this paper lie in optimizing the first two modules.

\subsection{Effect of Segment Length}

% \begin{figure}[t]
% \centerline{\includegraphics[width=\linewidth]{isit_fig3.2.png}} % Adjusts the width to fit within the text width
% \caption{Number of bits transmitted for both the original and improved protocols with a deletion rate of $\beta = 0.01$ across varying segment length coefficients, where the segment length is scaled by the coefficient.}
% \label{fig5}
% \end{figure}

\begin{figure}[t]
\centerline{\includegraphics[width=0.8\linewidth]{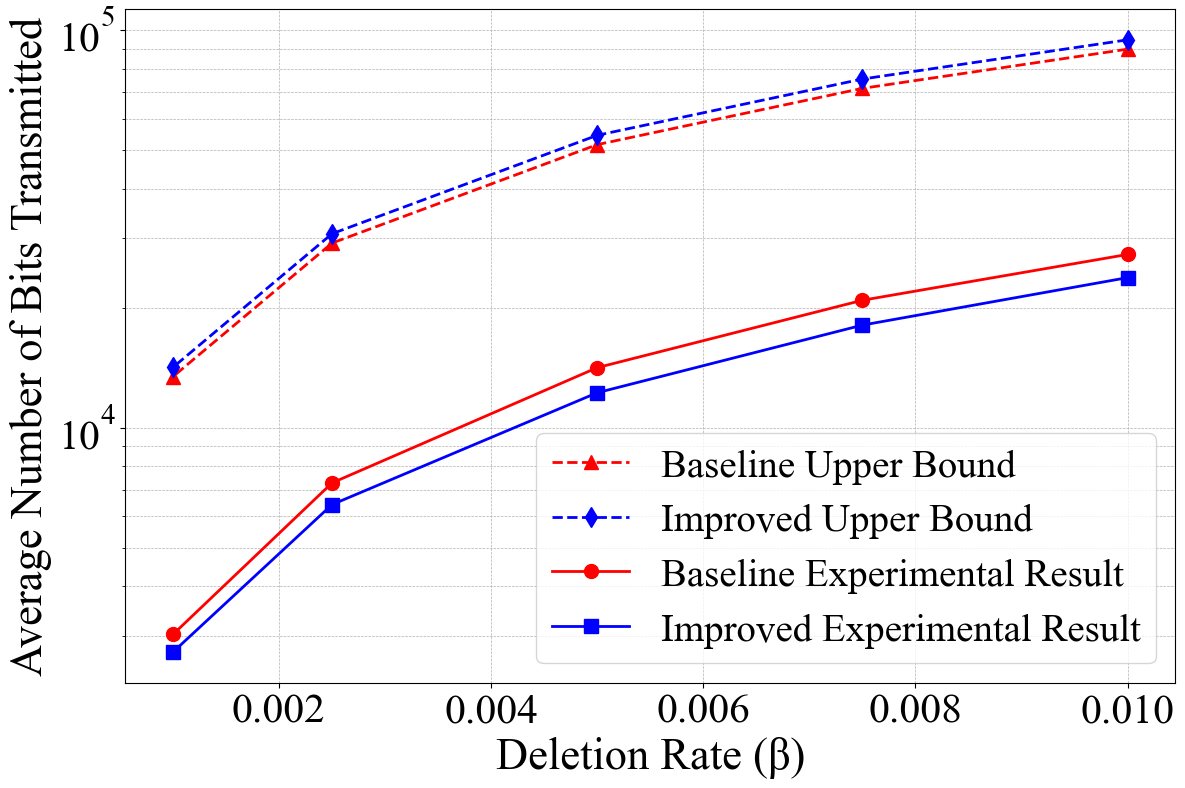}} % Adjusts the width to fit within the text width
\caption{Number of bits transmitted for the baseline and improved protocols at different deletion rates, with sequence length $n=50,000$ and segment length $L_S = \frac{2}{\beta}$.}
\vspace{-0.4cm}
\label{fig3}
\end{figure}

Figures \ref{fig5} and \ref{fig6} show how the number of bits transmitted varies with segment length for both protocols at deletion rates of \( \beta = 0.01 \) and \( \beta = 0.001 \), respectively. \nrev{The results reveal a non-monotonic trend: the number of bits first decreases and then increases as segment length grows. While the theoretical upper bound in Fig. \ref{r_s} predicts monotonic redundancy reduction with increasing segment length, experimental results suggest that the bound in Theorem \ref{theorem:improved} may not fully capture the effects of large segment lengths. Future work can focus on addressing this gap in order to provide more accurate analysis of the protocol.}

\nrev{The baseline protocol in \cite{ori} uses a segment length of \( L_S = \frac{1}{\beta} \), based on an expected single deletion per segment. However, many sections contain no deletions, indicating this choice is suboptimal. Increasing segment length reduces the number of pivots required for transmission while keeping deletions manageable within each section. For the parameters tested, the optimal segment length appears to lie between \( \frac{2}{\beta} \) and \( \frac{3}{\beta} \), highlighting the need for further exploration to determine optimal values for both protocols.}

\subsection{Comparison with the Baseline Protocol}

Figure \ref{fig3} compares the number of bits transmitted by the baseline and improved protocols for a sequence length \( n = 50{,}000 \) and a deletion rate \( \beta \) ranging from \( 0.001 \) to \( 0.01 \), with segment length \( L_S = \frac{2}{\beta} \) for both. \nrev{The results show that the improved protocol consistently reduces transmitted bits across all deletion rates, achieving up to a \( 13\% \) reduction compared to the baseline. While the baseline’s theoretical upper bound is slightly lower due to suboptimality in the two-deletion correction code, the improved protocol outperforms in practice, with potential for further reductions through enhanced multi-deletion correction codes.}

% \begin{figure}[t]
% \centerline{\includegraphics[width=\linewidth]{isit_fig2.2.png}} % Adjusts the width to fit within the text width
% \caption{Number of rounds for the baseline and improved protocols at different deletion rates, with sequence length $n=50,000$ and segment length $L_S = \frac{2}{\beta}$.}
% \label{fig4}
% \end{figure}

% Moreover, Figure \ref{fig3} shows that both protocols perform significantly better as the deletion rate decreases. For instance, the improved protocol requires approximately $3,000$ bits to correct deletions when $\beta = 0.001$, whereas significantly more bits are needed when $\beta = 0.01$. The increased inefficiency at higher deletion rates is likely due to a higher probability of delimiters with deletions during the divide-and-conquer process. In such cases, the encoder must repeatedly send delimiters until one is found in the deleted sequence, resulting in higher communication costs.

\section*{Conclusion}
This paper presented an improved deterministic communication protocol for synchronizing data files with random deletions, utilizing advanced multi-deletion correction codes and recursive partitioning techniques. By generalizing the expected communication cost for a broad class of deletion correction codes, we established the protocol's theoretical robustness and practical effectiveness, demonstrating reduced redundancy compared to the baseline protocol in both theoretical and experimental evaluations. Future work could explore the impact of the improved protocol on communication rounds and extend its applicability to hybrid error models involving insertions and substitutions, paving the way for enhanced performance in diverse data synchronization scenarios.
% \todo{comment, resolved: Future work will focus on other considerations such as number of rounds}

\section*{Acknowledgment}
\nrev{The authors acknowledge the support of NSF through grant CCF-CIF 2312872 for this work.}

\cleardoublepage
\bibliographystyle{ieeetr}
\bibliography{reference}

\cleardoublepage
\appendix

This appendix is organized into five subsections.

Appendix \ref{app:m0} establishes key theorems that provide upper bounds on the expected number of bits transmitted in each module, forming the foundation for the subsequent analysis. Appendices \ref{app:m1}, \ref{app:m2}, and \ref{app:m3} analyze the number of bits transmitted in the Matching Module, Deletion Recovery Module, and Error Correction Module, respectively. Additionally, Appendix \ref{app:m3} examines the relationship between pivot length and the probability of error in the first two modules. Finally, Appendix \ref{app:pivot length proof} presents the proof of Theorem \ref{theorem:pivot length}, a key result used in Appendix \ref{app:m3} to bound the number of bits transmitted in the Error Correction Module.  

Together, these results establish an upper bound on the total number of bits transmitted in the protocol, thereby proving Theorem \ref{theorem:improved}.

\subsection{Proof of Theorem \ref{theorem:improved}}
\label{app:m0}

To establish the total number of bits transmitted in the protocol, we analyze the contributions of each of the three modules separately. By summing these contributions, we obtain an upper bound on the overall redundancy. The following theorems provide key results for bounding the expected number of transmitted bits in each module:

\begin{theorem}
    \label{theorem:en1}
    Let \(s>0\) be the segment length multiplier. The expected number of bits transmitted in the Matching Module, denoted as \( \E{N_I} \), is upper-bounded by:
    \begin{equation}
        \E{N_I} \leq \frac{2}{s}n\beta\log\frac{1}{\beta} + o(n\beta\log \frac{1}{\beta}).
    \end{equation}    
\end{theorem}

\begin{theorem}
    \label{theorem:en2}
    Let \( s > 0 \) be the segment length multiplier. Let \( c > 0 \) be the delimiter length coefficient in the Deletion Recovery Module. Let \( w \) be the maximum number of deletions that can be corrected using a set of multi-deletion correction codes. Let \( a_i \geq 1 \) for \( i \in [1, \dots, w] \) denote the efficiency of the \( i \)-deletion correction code, where correcting \( i \) deletions in a sequence of length \( q \) requires \( i a_i \log q \) bits. Let \( a = \max\{a_1, a_2, \ldots, a_w\} \). The expected number of bits transmitted in the Deletion Recovery Module, denoted as \( \E{N_{II}} \), is upper-bounded by:
    \begin{equation}
        \E{N_{II}} \leq 2 \frac{s+1}{s} \left( \frac{2^w}{2^w-1}c + a \right) n \beta \log \frac{1}{\beta} + o(n \beta \log \frac{1}{\beta})
    \end{equation}
\end{theorem}

\begin{theorem}
    \label{theorem:en3}
    The expected number of bits transmitted in the Error Correction Module, denoted as \( \E{N_{III}} \), is upper-bounded by:
    \begin{equation}
        \E{N_{III}} \leq 2n\beta\log\frac{1}{\beta} + o(n\beta\log\frac{1}{\beta})
    \end{equation}
\end{theorem}

% \begin{lemma}
%     \label{lemma:separate}
%     Let \( s > 0 \) be the segment length multiplier. Let \( c > 0 \) be the delimiter length coefficient in the Deletion Recovery Module. Let \( w \) be the maximum number of deletions that can be corrected using a set of multi-deletion correction codes. Let \( a_i \geq 1 \) for \( i \in [1, \dots, w] \) denote the efficiency of the \( i \)-deletion correction code, where correcting \( i \) deletions in a sequence of length \( q \) requires \( i a_i \log q \) bits. Define \( a = \max\{a_1, a_2, \ldots, a_w\} \). The expected number of bits transmitted in the Matching Module \( \E{N_I}\), the Deletion Recovery Module \( \E{N_{II}}\), and the Error Correction Module \( \E{N_{III}} \) are respectively upper-bounded as follows:
%     \begin{equation}
%         \begin{split}
%             \E{N_I} &\leq \frac{2}{s} n\beta\log\frac{1}{\beta} + o(n\beta\log \frac{1}{\beta}),\\
%             \E{N_{II}} 
%             &\leq 2 \frac{s+1}{s} \left( \frac{2^w}{2^w-1}c + a \right) n \beta \log \frac{1}{\beta} + o(n \beta \log \frac{1}{\beta}),\\
%             \E{N_{III}} &\leq 2n\beta\log\frac{1}{\beta} + o(n\beta\log\frac{1}{\beta}).
%         \end{split}
%     \end{equation}
% \end{lemma}

Appendices \ref{app:m1}, \ref{app:m2}, and \ref{app:m3} will be devoted to proving the upper bounds stated in Theorem \ref{theorem:en1}, \ref{theorem:en2} and \ref{theorem:en3}, respectively.

Once the number of bits transmitted in the three modules is estimated, the total number of bits transmitted in the protocol can be bounded as:
\begin{equation}
    \begin{split}
        &\E{N_{I}} + \E{N_{II}} + \E{N_{III}} 
        \\&\leq \frac{2}{s} n \beta \log \frac{1}{\beta} + 2 \frac{s+1}{s} \left( \frac{2^w}{2^w-1}c + a \right) n \beta \log \frac{1}{\beta}  \\ 
        &\quad  +2 n \beta \log \frac{1}{\beta} + o(n \beta \log \frac{1}{\beta}) \\
        &= 2 \frac{s+1}{s} \left( \frac{2^w}{2^w-1}c + a + 1 \right) n \beta \log \frac{1}{\beta} + o(n \beta \log \frac{1}{\beta}) \\
        &\leq 2 \frac{s+1}{s} \left( \frac{2^w}{2^w-1}c + a + 2 \right) n \beta \log \frac{1}{\beta}.
    \end{split}
\end{equation}

This result directly leads to the bound established in Theorem \ref{theorem:improved}.

\subsection{Number of Bits Transmitted in Modules I}
\label{app:m1}
This appendix focuses on estimating the number of bits transmitted in the Matching Module. In the first module, bits are transmitted in two stages. First, the encoder sends pivot strings to the decoder. Then the decoder provides one-bit feedback per pivot to indicate whether each pivot was selected.  

Recall that the pivot length \(L_P\) need to satisfy the constraint  
\[
L_P \geq 3s + 8 + 2 \log \frac{1}{\beta},
\]  
where \(\beta\) is the deletion rate of the sequence, and \(s > 0\) is the segment length multiplier. A detailed explanation of this lower bound, along with its proof, can be found in Appendices \ref{app:m3} and \ref{app:pivot length proof} of this appendix, respectively.

By ensuring that this pivot length constraint holds and setting the segment length to \( L_S = \frac{s}{\beta} \), the number of segments \( k \) for a sequence of length \( n \) is upper-bounded by:  

\begin{equation}
    \begin{split}
        k &= \frac{n + L_P}{L_S + L_P} \leq \frac{n + L_P}{L_S} \\
        &= \frac{n\beta}{s} + (3 + \frac{8}{s})\beta + \frac{2}{s}\beta\log \frac{1}{\beta} 
        = \frac{n\beta}{s} +o(1),
    \end{split}
\end{equation}  
where \( o(1) \) is in terms of \( n \).  

Since there are \( k \) segments, there are \( k-1 \) pivots in the sequence. Given \( k \), the number of bits transmitted in the Matching Module is upper-bounded by:

\begin{equation}
    \E{N_I} = (k-1)L_P + k-1 
    \leq \frac{2}{s}n\beta\log\frac{1}{\beta} + o(n\beta\log \frac{1}{\beta}).
\end{equation}

Therefore, Theorem \ref{theorem:en1} is proven.

\subsection{Number of Bits Transmitted in Module II}
\label{app:m2}

This appendix focuses on deriving the number of bits transmitted in the Deletion Recovery Module. It is divided into two subsections. The first subsection analyzes a single section of length \( n_s \) with \( t \) deletions. The second subsection extends this analysis to multiple sections with varying lengths and deletion counts, ultimately establishing an upper bound on the total number of bits transmitted in the entire Deletion Recovery Module.

\subsubsection{Number of Bits Transmitted in a Single Section of Module II}

For a single section, the derivation separates the contributions of the encoder and decoder. Specifically, the number of bits transmitted by the encoder is divided into two components: the bits sent as delimiters to partition the sequence and the bits sent to correct the deletions. The number of bits transmitted by the decoder is inherently linked to the delimiters. Thus, the problem of bounding the number of bits transmitted in a single section reduces to bounding the contributions of the delimiters and deletion correction codes.

Let \( \E{N_{s_2}(t)} \) denote the total expected number of bits transmitted in a section over \( t \) deletions, with the length of the section being \( n_s \). Define \( \E{N_{AB}(t)} \) as the expected number of bits sent from the encoder to the decoder and \( \E{N_{BA}(t)} \) as the feedback bits sent from the decoder. Then, clearly, we have:

\begin{equation}
    \E{N_{s_2}(t)} = \E{N_{AB}(t)} + \E{N_{BA}(t)}.
\end{equation}

The term \( \E{N_{AB}(t)} \) can be further divided into \( \E{N_c(t)} \), the expected number of bits used for the delimiters in the divide-and-conquer approach, and \( \E{N_d(t)} \), the expected number of redundancy bits required to correct deletions. Thus, we have:

\begin{equation}
    \E{N_{AB}(t)} = \E{N_c(t)} + \E{N_d(t)}.
\end{equation}

Recall that in the Deletion Correction Module, the decoder only sends bits to the encoder to inform it of the result of the dividing process. Thus, \( \E{N_{BA}(t)} \) is related to \( \E{N_c(t)} \) only. Each time a delimiter is sent from the encoder to the decoder, the decoder needs to send several bits of information back to the encoder, informing it of the current case and what is required for the next step. One such case occurs when the delimiter is not found. In the case where the delimiter is found, the decoder will know how many deletions there are for each side. The decoder will then send a string of bits back to the encoder to inform it of the current situation.

Recall that we are using a set of deletion correction codes that can correct up to \( w \) deletions. After division, one side could have \( 0 \) to \( w \) deletions, or it may have more than \( w \) deletions. Thus, for each side, there are \( w+2 \) possible situations. Therefore, there are at most \( (w+2)^2 \) cases when the delimiter is found and both sides are considered.

Notice that the number of cases when the delimiter is found could actually be restricted to \( (w+2)^2 - 1 \), as the case where both sides have no deletions would not be possible, because the dividing process would not have started in the first place if there were no deletions. Thus, when we include the case where the delimiter is not found, there are at most \( (w+2)^2 \) possible cases.

Therefore, the following lemma provides an upper bound on the number of bits the decoder needs to send back to the encoder after receiving a delimiter:

\begin{lemma}
\label{lemma:ba_bits}
Let \( w>0 \) be the maximum number of deletions that can be corrected by the deletion correction code. Let \(l\) be the length of the delimiter. We have that:

\[
\E{N_{BA}(t)} \leq 2\log(w+2) \cdot \frac{\E{N_c(t)}}{l}.
\]
\end{lemma}

The term \( 2\log(w+2) \) represents the maximum number of bits the decoder transmits in a single communication round, while \( \frac{\E{N_c(t)}}{l} \) denotes the expected number of communication rounds required to divide the section.

After expressing both $\E{N_{AB}(t)}$ and $\E{N_{BA}(t)}$ in terms of $\E{N_c(t)}$ and $\E{N_d(t)}$, the next goal is to determine the actual expressions for these terms. Both will be estimated using induction.

We start by estimating $\E{N_c(t)}$. Clearly, when there are fewer than $w+1$ deletions, the current section can be immediately corrected without sending more delimiters. Therefore,

\begin{equation}
    \E{N_c(t)} = 0 \quad \text{for} \quad t = 0, 1, 2, \dots, w.
\end{equation}

Next, we consider the scenario where there are \( w+1 \) deletions in the section. This scenario illustrates how to analyze the problem when the number of deletions exceeds \( w \). In this case, the encoder sends a delimiter to the decoder to initiate the divide-and-conquer process. Once the first delimiter is located, two possible cases arise. 

In the first case, each half of the section contains at most \( w \) deletions and can be corrected using deletion correction codes. In the second case, all deletions occur within one half of the section, meaning the problem remains a \( (w+1) \)-deletion correction problem. The probability of the second case occurring is given by \( 2 \cdot \frac{1}{2^{w+1}} = \frac{1}{2^w} \). Thus, for \( \E{N_c(w+1)} \), we derive the following equation:

\begin{equation}
    \E{N_c(w+1)} \leq l + \frac{1}{2^w} \E{N_c(w+1)} + \left( 1 - \frac{1}{2^w} \right) \cdot 0.
\end{equation}

Since the substrings are smaller after splitting, the expected number of bits transmitted as delimiters in a substring with \( w+1 \) deletions can be upper-bounded by the original value of \( \E{N_c(w+1)} \). Therefore, the result of the equation provides an upper bound for \( \E{N_c(w+1)} \). Simplifying the equation, we get:

\begin{equation}
    \E{N_c(w+1)} \leq \frac{2^w}{2^w - 1} l.
\end{equation}

Given this equation, we now prove the following lemma:
\begin{lemma}
    \label{lemma:enc}
    Let \( w>0 \) be the maximum number of deletions that can be corrected by the deletion correction codes. Let \(l\) be the length of the delimiter. For a section with $t$ deletions, the expected number of bits sent as delimiters, $\E{N_c(t)}$, satisfies:

\begin{equation}
    \E{N_c(t)} \leq \frac{2^w}{2^w - 1} (t - 1) l.
\label{eq:ENc}
\end{equation}
\end{lemma}

\begin{proof}
We prove Lemma \ref{lemma:enc} by induction.

From the discussion above, it can be observed that Lemma \ref{lemma:enc} holds for all $t \leq w+1$. Now assume Lemma \ref{lemma:enc} holds for some integer $v > w+1$. We will show that Lemma \ref{lemma:enc} holds for $t = v+1$ as well. Similar to the previous analysis, when there are $t = v+1$ deletions in a section, after one division, the probability that there are $j$ deletions in the left half and $v+1-j$ deletions in the right half is given by

\begin{equation*}
    \frac{1}{2^{v+1}} \binom{v+1}{j}.
\end{equation*}

After sending a delimiter, with probability \( \frac{1}{2^{v+1}} \binom{v+1}{j} \), the original \( (v+1) \)-deletion correction problem reduces to two smaller subproblems: a \( j \)-deletion correction problem and a \( (v+1-j) \)-deletion correction problem. \rev{By recursively breaking the problem into smaller subproblems, each of which is easier to analyze, we can eventually establish an upper bound by bounding all the subproblems.} Additionally, notice that all deletions can fall into one half of the section with probability \( \frac{1}{2^v} \), similar to what has been discussed previously, causing the problem to remain a \( (v+1) \)-deletion correction problem.

Given this analysis, the following recurrence holds:  

\begin{equation}
    \begin{split}
        \label{inc}
        \E{N_c&(v+1)} \leq l + \frac{1}{2^v} \E{N_c(v+1)} \\ 
        &+ \sum_{j=1}^{v} \frac{1}{2^{v+1}} \binom{v+1}{j} (\E{N_c(j)} + \E{N_c(v+1-j)}).
    \end{split}
\end{equation}

We can further simplify $\E{N_c(v+1)}$ as follows:

\begin{equation}
    \begin{split}
        &\E{N_c(v+1)} \\&\stackrel{(a)}{\leq} \frac{\left(l + \sum_{j=1}^{v} \frac{\binom{v+1}{j}}{2^{v+1}}  (\E{N_c(j)} + \E{N_c(v+1-j)})\right)}{1 - 2^{-v}} \\
        &\stackrel{(b)}{\leq} \frac{l + \sum_{j=1}^{v} \frac{\binom{v+1}{j}}{2^{v+1}}  (\frac{2^w}{2^w - 1} (j-1) + \frac{2^w}{2^w - 1} ((v+1-j)-1))l}{1 - 2^{-v}} \\
        &\stackrel{(c)}{=} \frac{l + \sum_{j=1}^{v} \frac{1}{2^{v+1}} \binom{v+1}{j} \frac{2^w}{2^w - 1} (v-1) l}{1 - 2^{-v}} \\
        &\stackrel{(d)}{=} \frac{l}{1 - 2^{-v}} + \frac{2^w}{2^w - 1} (v-1) l \\
        &\stackrel{(e)}{<} \frac{2^w}{2^w - 1} ((v+1) - 1) l.
    \end{split}
\end{equation}

Here (a) follows from applying the recurrence in Equation (\ref{inc}), which expresses \( \E{N_c(v+1)} \) in terms of the expected values of \( N_c \) at smaller indices. (b) leverages the inductive hypothesis that $\E{N_c(t)}  \leq \frac{2^w}{2^w - 1} (t - 1) l$ for $t \leq v$. (c) factors out the common term from the summation. (d) simplifies the sum \( \sum_{j=1}^{v} \binom{v+1}{j} = 2^{v+1} - 2 \), using the binomial sum identity. (e) is based on the fact that \( \frac{1}{1 - 2^{-v}}<\frac{2^w}{2^w - 1}\) with \(w<v\).

% \begin{equation}
%     \begin{split}
%         EN_c(k+1) &= \frac{1}{1 - 2^{-k}} \left(l \sum_{j=1}^{k} \frac{1}{2^{k+1}} \binom{k+1}{j} EN_c(j)  \\
%         &+ \sum_{j=1}^{k}\frac{1}{2^{k+1}}\binom{k+1}{j}\mathbb{E}[N_c(k+1-j)] \right) \\
%         &\leq \frac{l + \sum_{j=1}^{k} \frac{1}{2^{k+1}} \binom{k+1}{j} (\frac{2^w}{2^w - 1} (j-1) l + \frac{2^w}{2^w - 1} ((k+1-j)-1) l)}{1 - 2^{-k}} \\
%         &= \frac{l + \sum_{j=1}^{k} \frac{1}{2^{k+1}} \binom{k+1}{j} \frac{2^w}{2^w - 1} (k-1) l}{1 - 2^{-k}} \\
%         &= \frac{l}{1 - 2^{-k}} + \frac{2^w}{2^w - 1} (k-1) l \\
%         &< \frac{2^w}{2^w - 1} ((k+1) - 1) l.
%     \end{split}
% \end{equation}

% \begin{align}
%     EN_c(k+1) &= \frac{l + \sum_{j=1}^{k} \frac{1}{2^{k+1}} \binom{k+1}{j} (EN_c(j) + EN_c(k+1-j))}{1 - 2^{-k}} \\
%     &\leq \frac{l + \sum_{j=1}^{k} \frac{1}{2^{k+1}} \binom{k+1}{j} (\frac{2^w}{2^w - 1} (j-1) l + \frac{2^w}{2^w - 1} ((k+1-j)-1) l)}{1 - 2^{-k}} \\
%     &= \frac{1}{1 - 2^{-k}} l + \sum_{j=1}^{k} \frac{1}{2^{k+1}} \binom{k+1}{j} \frac{2^w}{2^w - 1} (k-1) l}{1 - 2^{-k}} \\
%     &= \frac{l}{1 - 2^{-k}} + \frac{2^w}{2^w - 1} (k-1) l \\
%     &< \frac{2^w}{2^w - 1} ((k+1) - 1) l.
% \end{equation}

Thus, it is proven that $\E{N_c(t)} \leq \frac{2^w}{2^w - 1} (t - 1) l$ holds for $t = v+1$. Therefore, Lemma \ref{lemma:enc} is proven by induction.
\end{proof}

Now recall that the length of the delimiter is determined by the equation \(l = c \log n_s\), where $c$ is a changeable parameter. The bound for $\E{N_c(t)}$ can therefore be expressed as:

\begin{equation}
    \E{N_c(t)} \leq \frac{2^w}{2^w - 1} (t - 1) c \log n_s < \frac{2^w}{2^w - 1} c t \log n_s.
\end{equation}

Next, we shall determine a bound for $\E{N_d(t)}$, which represents the average number of bits transmitted to actually correct the deletions in a section.

First, consider the case where the number of deletions in a section is no more than \(w\). Recall that it is stated in Theorem \ref{theorem:improved} that we use \(ia_i\log q\) bits to correct \(i\) deletions in a sequence of length \(q\) when \(i\leq w\). Therefore, $\E{N_d(t)}$ is given by:

\begin{equation}
    \E{N_d(t)} \leq a_t t \log n_s \text{ for } t\leq w.
\end{equation}

Now when the number of deletions in a section exceeds \( w \), the section must be divided before it can be corrected. In this case, the protocol will require the use of a combination of deletion correction codes to handle all the deletions, rather than relying on a single code. 

As previously derived for \( \E{N_c(t)}\), after one division, the deletion correction problem is broken down into two subproblems, each with no more deletions than before. With probability \( \frac{1}{2^t} \binom{t}{j} \), after one division, a section with \( t \) deletions is split such that the left half contains \( j \) deletions and the right half contains \( t-j \) deletions. Applying a similar reasoning as before, we obtain the following equation for \( \E{N_d(t)} \):
\begin{equation}
    \label{eq:ind}
    \begin{split}
        \E{N_d(t)} &\leq \frac{1}{2^{t-1}} \E{N_d(t)} \\&\quad+ \sum_{j=1}^{t-1} \left[ \frac{\binom{t}{j}}{2^t} \left( \E{N_d(j)} + \E{N_d(t-j)} \right) \right]
    \end{split}
\end{equation}

Recall that \( a = \max \{ a_1, a_2, \dots, a_w \} \), we then have the following lemma:

\begin{lemma}
    \label{lemma:end}
    For a section with \( t \) deletions, the expected number of bits sent as redundancy for deletion correction \( \E{N_d(t)} \) satisfies:
    \begin{equation}
        \E{N_d(t)} \leq a t \log n_s
    \label{eq:ENd}
    \end{equation}
\end{lemma}

\begin{proof}
    We prove this lemma by induction.\\
Since \( a \) is defined as the maximum over all efficiencies of the deletion correction codes, it follows that Lemma \ref{lemma:end} holds for \( t \leq w \), as:
\begin{equation}\label{eq:nd_upper}
    \E{N_d(t)} = a_t t \log n_s \leq a t \log n_s \quad \text{for} \quad t \leq w.
\end{equation}

Now, assume that Lemma \ref{lemma:end} holds for \( t = v \). Then, for \( t = v + 1 \), we have:  

\begin{equation}
    \begin{split}
        &\E{N_d(v+1)} \stackrel{(a)}{\leq}  \left( 1 + \frac{1}{2^v - 1} \right) \times \\
        &\quad\sum_{j=1}^{v} \left[ \frac{\binom{v+1}{j}}{2^{v+1}}\left( \E{N_d(j)} + \E{N_d(v+1-j)} \right) \right] \\
        &\stackrel{(b)}{\leq} \left( 1 + \frac{1}{2^v - 1} \right) \sum_{j=1}^{v} \left[ \frac{\binom{v+1}{j}}{2^{v+1}} a (v+1) \log n_s \right] \\
        &\stackrel{(c)}{=} \left( 1 + \frac{1}{2^v - 1} \right) a (v+1) \log n_s \frac{\sum_{j=1}^{v} \binom{v+1}{j}}{2^{v+1}} \\
        &\stackrel{(d)}{=} \left( 1 + \frac{1}{2^v - 1} \right) a (v+1) \log n_s \frac{2^{v+1} - 2}{2^{v+1}} \\
        &\stackrel{(e)}{=} \frac{2^{v+1}}{2^{v+1} - 2} \frac{2^{v+1} - 2}{2^{v+1}} a (v+1) \log n_s \\
        &\stackrel{(f)}{=} a (v+1) \log n_s
    \end{split}
\end{equation}  

Here (a) follows from applying the recurrence in Equation (\ref{eq:ind}). (b) leverages the inductive hypothesis, which assumes that \( \E{N_d(j)} \leq a j \log n_s \) for all \( j \leq v \), allowing us to bound each term in the summation. (c) factors out the common term from the summation. (d) simplifies the sum \( \sum_{j=1}^{v} \binom{v+1}{j} = 2^{v+1} - 2 \), using the binomial sum identity. (e) rewrites the fraction to show cancellation explicitly. (f) follows from canceling out the first two terms  in previous expression, completing the induction step.  

Thus, the bound is valid for \( t = v + 1 \). Therefore, Lemma \ref{lemma:end} is proven.
\end{proof}

With the bounds for \( \E{N_c(t)} \) and \( \E{N_d(t)} \) derived in Lemma \ref{lemma:enc} and \ref{lemma:end}, respectively, we can also bound \( \E{N_{AB}(t)} \) and \( \E{N_{BA}(t)} \) as follows:
\begin{equation}
    \begin{split}
        \E{N_{AB}(t)} &= \E{N_c(t)} + \E{N_d(t)}\\ &\leq \left( \frac{2^w}{2^w - 1} c + a \right) t \log n_s \\
        \E{N_{BA}(t)} &\leq 2 \log(w+2) \frac{\E{N_c(t)}}{l} \\ &\leq \frac{2^{w+1} \log(w+2)}{2^w - 1} (t - 1)
    \end{split}
\end{equation}

Therefore, the expected number of bits required to transmit to correct a section with length \( n_s \) and deletion \( t \) is:
\begin{equation}
    \label{eq:ens}
    \begin{split}
        \E{&N_{s_2}(t)} = \E{N_{AB}(t)} + \E{N_{BA}(t)} \\&\leq \left( \frac{2^w}{2^w - 1} c + a \right) t \log n_s + \frac{2^{w+1} \log(w+2)}{2^w - 1} (t - 1).
    \end{split}
\end{equation}
\subsubsection{Total Number of Bits Transmitted in Module II}
Note that $\E{N_{s_2}(t)}$ represents the expected number of bits required to correct a \textit{single} section. However, the entire sequence is divided into multiple sections, each with a different length and a varying number of deletions. 

Let the $i$th section be denoted as $F_i$, $|F_i|$ be the length of the section, and $\delta_i$ is the number of deletions within that section. Given Equation (\ref{eq:ens}), the expected number of bits transmitted for the $i$th section can be expressed as
\begin{equation}
    \left(\frac{2^w}{2^w-1}c+a\right)\delta_i \log |F_i| + \frac{2^{w+1} \log(w+2)}{2^w-1}(\delta_i-1).
\end{equation}

Let \(k'\) denote the number of sections for the entire sequence. Since the total number of bits transmitted is the sum of the bits transmitted for each individual section, the total number of bits transmitted in the Deletion Recovery Module, denoted as $\E{N_{II}}$, is upper-bounded by
\begin{equation}
    \label{eq:n2_sum}
    \begin{split}
        &\mathbb{E}\left[\sum_{i=1}^{k'}  \left(\frac{2^w}{2^w-1}c+a\right)\delta_i \log |F_i|\right] \\
        &+ \mathbb{E}\left[\frac{2^{w+1} \cdot  \log(w+2)}{2^w-1}\sum_{i=1}^{k'}(\delta_i-1)  \right] .
    \end{split}
\end{equation}

% \begin{equation}
%     \begin{split}
%         &\E{N_{II}} \\
%         &= \mathbb{E}\left[\sum_{i=1}^{k'} \left( \underbrace{ \left(\frac{2^w}{2^w-1}c+a\right)\delta_i \log |F_i|}_{a} + \frac{2^w \cdot 2 \log(w+2)}{2^w-1}(\delta_i-1) \right) \right] .
%     \end{split}
% \end{equation}
To calculate the second term in Equation (\ref{eq:n2_sum}), we note that the expected number of deletions across the entire sequence is \(n\beta\) and, thus, we have the relation:
\begin{equation}
    \mathbb{E} \left[ \sum_{i=1}^{k'} \delta_i \right] = n \beta.
\end{equation}

The next lemma helps to evaluate the first term in Equation (\ref{eq:n2_sum}):
\begin{lemma}
    \label{lemma:edf}
    $\mathbb{E}[\delta_i \log |F_i|]$ is upper-bounded by
    \begin{equation}
        8(s+1) + 2(s+1) \log (s+1) + 2(s+1) \log \frac{1}{\beta}.
    \end{equation}
\end{lemma}

\begin{proof}
    In the Appendix of \cite{ori}, it is shown that
    \begin{equation}
        \E{\delta_i\log|F_i|} = \E{\beta |F_i|\log|F_i|}.
    \end{equation}
    Thus, the problem reduces to finding a bound for $\E{\beta |F_i|\log|F_i|}$.
    
    Recall that a section is the concatenation of several segments and pivots. Let the section $F_i$ consist of $r$ segments and $r-1$ pivots. The length of the section is calculated as 
    \[
    |F_i| = rL_S + (r-1)L_P \leq r(L_S + L_P).
    \]
    Therefore, we have the following upper bound for $\E{\beta |F_i|\log|F_i|}$:
    \begin{equation}
        \begin{split}
            \E{\beta |F_i|\log|F_i|} & \leq \E{\beta r(L_S + L_P) \log r(L_S + L_P)} \\
            & = \beta (L_S + L_P) \E{r \log r + r \log (L_S + L_P)}.
        \end{split}
    \end{equation}
    Since the length of the pivot satisfies $L_P \leq \frac{1}{\beta} = \frac{1}{s} L_S$, we have $\beta (L_S + L_P) \leq s + 1$. Additionally, we know that $r \log r \leq r^2$. Using these upper bounds, we can further express the upper bound for $\E{\beta |F_i|\log|F_i|}$ as
    \begin{equation}
        \E{\beta |F_i|\log|F_i|} \leq (s+1) \E{r^2} + (s+1) \log (L_S + L_P) \E{r}.
    \end{equation}

    Next, we know that the probability of finding a pivot is
    \begin{equation}
        p = \frac{k'-1}{k-1} = 1 - L_P \beta + 2\beta + o(\beta).
    \end{equation}
    Notice that $r$ follows a geometric distribution with
    \begin{equation}
        \Pr(|F_i| = rL_S + (r-1)L_P) = p(1-p)^{r-1},
    \end{equation}
    so
    \begin{equation}
        \begin{split}
            \E{r} & = \frac{1}{p}, \\
            \E{r^2} & = \text{Var}(r) + \E{r}^2 = \frac{2 - p}{p^2}.
        \end{split}
    \end{equation}
    Since $p \geq \frac{1}{2}$, we have
    \begin{equation}
        \begin{split}
            \E{r} & \leq 2, \\
            \E{r^2} & \leq 8.
        \end{split}
    \end{equation}
    Furthermore, it holds that
    \begin{equation}
        \log(L_S + L_P) \leq \log \frac{s+1}{s} L_S = \log(s+1) + \log \frac{1}{\beta}.
    \end{equation}
    Substituting these inequalities into the bound, we get
    \begin{equation}
        \E{\delta_i \log |F_i|} \leq 8(s+1) + 2(s+1) \log(s+1) + 2(s+1) \log \frac{1}{\beta}.
    \end{equation}
    This completes the proof of Lemma \ref{lemma:edf}.\\
\end{proof}

Using the previous equations, $\E{N_{II}}$ can be upper-bounded as follows:

\begin{equation}
    \begin{split}
        &\E{N_{II}}\\ 
        &\leq k' \left( \frac{2^w}{2^w-1}c+a \right) \mathbb{E} \left[ \delta_i \log |F_i| \right] + \frac{2^w \cdot 2 \log(w+2)}{2^w-1} n \beta \\
        &\stackrel{(a)}{\leq} \frac{n \beta}{s} \left( \frac{2^w}{2^w-1}c+a \right) (s+1)\left( 8 + 2 \log(s+1) + 2 \log \frac{1}{\beta} \right) \\
        & \quad + \frac{2^w \cdot 2 \log(w+2)}{2^w-1} n \beta \\
        &= 2 \frac{s+1}{s} \left( \frac{2^w}{2^w-1}c + a \right) n \beta \log \frac{1}{\beta} + o(n \beta \log \frac{1}{\beta}),
    \end{split}
\end{equation}
where $(a)$ is based on the fact that $k' \leq k = \frac{n \beta}{s}$. This finishes the proof of Theorem \ref{theorem:en2}.

\subsection{Number of Bits Transmitted in Module III}  
\label{app:m3}  

In this section, we analyze the number of bits transmitted in the Error Correction Module. Since this module is responsible for correcting errors from the first two modules, it is essential to first establish an upper bound on the error probability of the first two modules before estimating the number of bits transmitted in the Error Correction Module.

To estimate the probability of error in the first two modules, we analyze the error probability in each module separately. We first focus on the Matching Module. As discussed in Section~\ref{section:algorithm}, when the decoder identifies a pivot in its sequence, it is not necessarily the correct match for the corresponding pivot in the encoder. A false pivot may be incorrectly selected as the true match. Intuitively, increasing the pivot length reduces the likelihood of false pivot matches, as longer pivots are less likely to appear erroneously in the decoder’s sequence. The following theorem formalizes this relationship between pivot length and matching error probability.

\begin{theorem} \label{theorem:pivot length}
    Let \( s > 0 \) be the segment length multiplier. If the pivot length \( L_P \) satisfies 
    \begin{equation}
        \label{eq:pl}
        L_P \geq 3s+8+2\log \frac{1}{\beta},
    \end{equation}
    then, with probability at least \(1-2^{-\Omega(n)}\), the probability of a matching error for each pivot is at most \(\beta+o(\beta)\).
\end{theorem}

Proof of Theorem \ref{theorem:pivot length} can be found in Appendix \ref{app:pivot length proof}. Theorem \ref{theorem:pivot length} implies that by setting the pivot length to satisfy Equation \eqref{eq:pl}, the probability of incorrect pivot matching is upper-bounded by \(\beta+o(\beta)\) with high probability. Since each section is the common neighbor between two selected pivots, and the probability of an individual pivot mismatch is at most \(\beta+o(\beta)\), the probability that a section in the decoder does not correspond to a deleted version of its counterpart in the encoder is upper-bounded by \(2\beta+o(\beta)\). When a section in the decoder is not the deleted version of the corresponding section in the encoder after the Matching Module, the corrected sequence produced by the Deletion Recovery Module may deviate from the original encoder section, introducing errors that must be handled in the Error Correction Module.  

Next, we estimate the probability of error in the Deletion Recovery Module. Errors in this module arise due to mismatched delimiters in the divide-and-conquer approach. The following lemma provides an upper bound on the probability of such errors.

\begin{lemma} \label{theorem:delimiter length}  
    Let \( c > 0 \) be the delimiter length coefficient in the Deletion Recovery Module. For \( c=3 \), the probability of error when reconstructing a section in the Deletion Recovery Module is upper-bounded by \(\beta^3\log \frac{1}{\beta} = o(\beta)\).  
\end{lemma}  

Lemma~\ref{theorem:delimiter length} provides an upper bound on the probability of errors introduced by the Deletion Recovery Module. To prove this lemma, we refer to Theorem 1 in \cite{VT}, which bounds the probability of error in a deletion correction protocol using the divide-and-conquer approach combined with VT codes.  

\begin{theorem}[Theorem 1, \cite{VT}]  
\label{theorem:delimiter error}  
    Suppose there are \(d\) deletions in a sequence of length \(n\). Let \(l=c\log n\) be the length of the delimiter with fixed parameter \(c\). Then, the probability of error is at most \(\frac{d\log n}{2n^c}\).  
\end{theorem}  

In the divide-and-conquer approach, as similarly used in \cite{VT}, the only source of error is delimiter mismatch. Therefore, Theorem~\ref{theorem:delimiter error} directly bounds the probability of delimiter mismatch when the divide-and-conquer approach is applied with VT codes. In our improved protocol, multi-deletion correction codes are used instead of VT codes, which reduces the number of required divisions. Since the probability of a single delimiter mismatch remains the same while the number of splits decreases, the probability of delimiter mismatch is also reduced. Consequently, Theorem~\ref{theorem:delimiter error} serves as a valid upper bound on delimiter mismatch probability for our improved protocol.  

In our setting, we take \( c=3 \). The expected number of deletions is given by \(\mathbb{E}[d] = \frac{1}{p}\beta L_S = \frac{s}{p}\), and the expected substring length is \(\mathbb{E}[n] = \frac{1}{p}L_S = \frac{s}{p\beta}\). \rev{Recall that \(p=\frac{k'-1}{k-1} = 1 - L_P \beta + 2\beta + o(\beta)<1\).} Substituting these values into the upper bound from Theorem~\ref{theorem:delimiter error}, we obtain:  
\[
\frac{p^2\beta ^3}{2s^2}\log \frac{s}{p\beta}<\frac{\beta^3}{s^2} \log \frac{s}{p\beta} = o(\beta).
\]  
This establishes the desired error probability bound, thus proving Lemma~\ref{theorem:delimiter length}. 

Given Theorem~\ref{theorem:pivot length} and Lemma~\ref{theorem:delimiter length}, we derive the following corollary:

\begin{corollary}
\label{corollary:prob err}
    By setting the pivot length \( L_P \) to satisfy the constraint in eq. \eqref{eq:pl} and choosing the delimiter length coefficient as \( c=3 \), the probability that a section in the decoder, after passing through the first two modules, does not exactly match its corresponding section in the encoder is upper-bounded by \( 2\beta+o(\beta) \).
\end{corollary}

Corollary \ref{corollary:prob err} follows directly from applying the union bound to the probability of errors in the first two modules. This corollary establishes an upper bound on the probability that a section of the decoder does not perfectly match its corresponding section in the encoder after the first two modules. 

Since a section consists of multiple bits, and each bit has the same probability of being deleted, we assume that bit errors are uniformly distributed across the section. By permuting the bits within each section, we ensure that errors appear independent and identically distributed, making the error model closely resemble a binary symmetric channel. Consequently, the probability of an individual bit being incorrectly synchronized is also upper-bounded by \(2\beta + o(\beta)\). Given this bound on the bit error probability, the expected number of bits transmitted in the Error Correction Module is estimated in the following lemma:

\begin{lemma}
    \label{lemma:moduleIII}
    Let \( H(\cdot) \) denote the binary entropy function, defined as \( H(p) = p\log\frac{1}{p} + (1-p)\log\frac{1}{1-p}\) for \( 0 < p < 1 \). The expected number of bits transmitted in the Error Correction Module, denoted as \( \mathbb{E}[N_{III}] \), is upper-bounded by:
    \begin{equation}
       \mathbb{E}[N_{III}] = nH(2\beta + o(\beta)) = 2n\beta\log\frac{1}{\beta} + o(n\beta\log\frac{1}{\beta}).
    \end{equation}
\end{lemma}

The proof of Lemma~\ref{lemma:moduleIII} follows the approach used to estimate the upper bound on the number of bits transmitted in the third module of \cite{ori} and \cite{ori_ex}. Since the bit error probability has already been upper-bounded using the extended result of Corollary~\ref{corollary:prob err}, we apply the capacity of a binary symmetric substitution channel to derive an upper bound on the number of bits transmitted in the Error Correction Module, thereby concluding the proof of Theorem \ref{theorem:en3}.

\subsection{Proof of Theorem \ref{theorem:pivot length}}
\label{app:pivot length proof}

The proof of Theorem \ref{theorem:pivot length} closely follows that of Theorem 5 in \cite{ori}, which establishes the relationship between pivot length and the probability of matching error for the baseline protocol. The key difference in the improved protocol is the modification of the segment length: while the baseline protocol sets \( L_S = \frac{1}{\beta} \), the improved protocol uses \( L_S = \frac{s}{\beta} \). Consequently, the expected number of deletions per segment changes from 1 to \( s \). This section outlines how this change impacts the derivation, highlighting the key modifications while omitting redundant calculations. Given the strong similarity between the two proofs, this section provides a high-level sketch of how the segment length constraint is derived, emphasizing the differences from the baseline case. For the full original proof, please refer to Theorem 5 in \cite{ori}.  

In essence, Theorem \ref{theorem:pivot length} is proven by demonstrating that, with high probability, the fraction of false pivots among the pivots selected by the decoder is at most \(\beta + o(\beta)\). To clarify the proof structure, we first revisit the matching module with additional operational details. As described in Section \ref{section:algorithm}, the decoder attempts to match the pivots sent by the encoder within its own sequence. Due to random deletions in the decoder's sequence, not all pivots have a correct match that can be found by the decoder. However, with high probability, a significant fraction of pivots will have at least one correct match. This observation is formalized in the following lemma. 

\begin{lemma}
    \label{lemma:correct_match}
    Let \( X \) be the sequence of the encoder and \( Y \) be the sequence of the decoder. Suppose the length of \( X \) is \( n \), the deletion rate is \(\beta\), and the length of each pivot is \( L_P \). Define \( R = 1 - L_P \beta + 2\beta \). Then, for a random binary string \( X \) with \(k-1\) pivots, with probability \( 1 - 2^{-\Omega(n)} \), there exist \( (R + o(\beta)) k \) pivots that have at least one correct match in \( Y \).
\end{lemma}

Lemma \ref{lemma:correct_match} is the improved protocol's counterpart to Lemma 6 in \cite{ori}. Its proof remains identical to that of Lemma 6 in \cite{ori}. \rev{The impact of the segment length adjustment is captured in the new choice of pivot length \( L_P \) in Lemma \ref{lemma:correct_match}.}

Notably, a given pivot may have two correct matches. This seemingly counterintuitive situation arises from the definition of a correct match. Following the definition in \cite{ori}, a correct match in our setting consists of two cases. If no deletion occurs in the pivot, the corresponding sequence in the decoder is trivially the correct match of the pivot in the encoder. When a single deletion occurs within the pivot, a sequence remains a correct match if it \rev{possesses the exact sequence of the pivot and} starts at the expected position or ends at the expected position. This is because the pivot sequence remains in place, and the deletion can be interpreted as occurring outside the pivot.  

For example, consider the sequence \(00010\), where the pivot is the three-bit subsequence \(001\). If the third bit of the sequence is deleted, the resulting sequence is \(0010\), with the pivot appearing as \(01\). However, this deletion pattern is equivalent to deleting the first bit instead, which results in the same sequence \(0010\) while preserving the pivot \(001\). Thus, the pivot is considered intact, and a correct match remains even after a deletion.  

A pivot has two correct matches when its string consists entirely of either zeros or ones, the immediate undeleted bits before and after the pivot share the same value as the pivot, and exactly one deletion occurs within the pivot. In this case, the deletion can be interpreted as occurring either within the pivot or just before/after it. From the decoder’s perspective, both interpretations are valid, allowing the pivot to be matched in two ways. \rev{An example of this scenario is as follows. Consider the sequence \(00000\), where the three center bits form the pivot, and the only deleted bit is the third bit of the sequence. The resulting sequence \(0000\) can also be obtained by interpreting the deleted bit as either the first or the last bit. In both cases, the pivot remains intact.} With high probability, the number of pivots exhibiting this phenomenon is \(o(\beta k)\). This result will be used in the subsequent proof.

Given that, with high probability, there are \( (R + o(\beta)) k \) pivots with at least one correct match, the decoder aims to identify \( (R + o(\beta)) k \) pivots as the selected pivots among all detected pivots. Notably, after attempting to locate the pivots based on the received strings from the encoder, the decoder is likely to find more than \( (R + o(\beta)) k \) pivot matches. However, some of these matches are false positives—erroneous matches arising due to the pivot string randomly appearing elsewhere in the sequence.  

To mitigate this issue, the decoder employs the matching-graph algorithm described in \cite{ori} to identify the most likely \( (R + o(\beta)) k \) correct pivots. This graph-based algorithm selects a set of pivots as potential candidates if they satisfy the constraints outlined in Equation (9) of \cite{ori}. These constraints ensure that the selected pivots maintain the correct cardinal order and that the number of bits between two pivots in the decoder does not exceed the corresponding distance in the encoder's sequence, given that all edits are deletions. Ultimately, the selected pivots are the subset of \( (R + o(\beta)) k \) pivots that meet these constraints. While these constraints do not guarantee that every selected pivot is correct, they ensure that, with high probability, the error rate in pivot matching remains asymptotically small, as stated in Theorem \ref{theorem:pivot length}.

Assume that among the \( (R + o(\beta)) k \) selected pivots, \(\alpha k\) pivots are false pivots resulting from mismatches. Theorem \ref{theorem:pivot length} is proven by analyzing the probability of having \(\alpha k\) false pivots within the selected pivots for every \(\beta < \alpha \leq R + o(\beta)\). By summing over all such probabilities, we establish that the total probability of selecting \(\alpha k\) false pivots with \(\alpha > \beta\) remains bounded and asymptotically small.

First, we consider the case where \( \beta < \alpha < \frac{1}{2} \). Given a pivot matching error rate of \( \alpha \), the set of selected pivots consists of \( (R - \alpha + o(\beta)) k \) correct pivots and \( \alpha k \) false pivots. To determine the overall probability of selecting \( \alpha k \) false pivots, we first compute the number of possible configurations in which correct and false pivots can appear in the sequence. This count is then multiplied by the probability of each configuration occurring.  

We begin by estimating the number of possible arrangements of correct and false pivots in the sequence. By Lemma \ref{lemma:correct_match}, \( (R + o(\beta)) k \) pivots have a correct match in the decoder’s sequence. The \((R - \alpha + o(\beta)) k\) correct pivots must be chosen from this set. Therefore, the number of possible selections of cardinal indices for the correct pivots is given by:
\[
\binom{Rk+o(\beta)k}{Rk-\alpha k+o(\beta)k},
\]
where the correct pivots are selected from all potential correctly matched pivots in the decoder’s sequence.

For most correct pivots, their starting positions are uniquely determined by their cardinal indices, as a correctly matched pivot typically has only one valid location. However, as previously noted, \( o(\beta k) \) pivots may have two valid matches in the decoder’s sequence. In these cases, either position can be chosen as the pivot’s starting location. Consequently, the total number of possible configurations for correct pivots \(N_c\) is:
\[
N_c=\binom{Rk+o(\beta)k}{Rk-\alpha k+o(\beta)k} \cdot 2^{o(\beta)k}.
\]

We now analyze the false pivots. The cardinal indices of the false pivots can be any number except those already selected as the cardinal indices of the correct pivots. Therefore, the number of possible choices for the cardinal indices of the false pivots is given by:
\[
\binom{(1-R)k+\alpha k +o(\beta)k}{\alpha k}.
\]

Unlike correct pivots, where the starting position is typically determined by the cardinal index, the actual starting positions of false pivots are not fixed. A false pivot’s starting position is valid as long as it satisfies the two constraints outlined in Equation (9) of \cite{ori}, which were previously discussed. Given a specific selection of cardinal indices for the false pivots, the number of possible false pivot configurations arising from different starting positions is upper-bounded by:
\[
\binom{2s\alpha k(1+\frac{1}{s}+\beta L_P-2\beta+\frac{o(\beta)}{s\alpha})}{2\alpha k} \leq 2^{(3s+2)\alpha k}.
\]
Here, \( s \) represents the segment length multiplier. This bound is a variant of Equation (13) in \cite{ori}, and its derivation follows the same approach. In essence, Equation (13) in \cite{ori} is derived by counting the number of valid ways to distribute deletions across different sections. Our derivation follows the same proof strategy, with the key difference being that in the improved protocol, the segment length is scaled by \( s \). Consequently, the average number of deleted bits per section also scales by \( s \), leading to the introduction of the segment length multiplier \( s \) in the upper bound.

Thus, the upper bound on the total number of possible configurations for false pivots, denoted as \( N_f \), is given by:
\[
N_f=\binom{(1-R)k+\alpha k +o(\beta)k}{\alpha k} \cdot 2^{(3s+2)\alpha k}.
\]

With the previous derivations, the total number of possible configurations for all correct and false pivots is given by:
\[
N_{\text{conf}} = N_c N_f.
\]
For each configuration to occur, a total of \( L_P \alpha k \) bits must be fixed to ensure the existence of false pivots at the designated positions. \rev{Given that the sequence of the decoder is generated by an i.i.d. Bernoulli source with parameter \(\frac{1}{2}\), and that different sections of the decoder sequence are independent,} the probability of this event occurring is \( 2^{-L_P\alpha k} \). Therefore, the overall probability of having \(\alpha k\) false pivots in the selected pivots, denoted as \( P_{\alpha1} \), is given by:
\[
P_{\alpha1} = N_{\text{conf}} \cdot 2^{-L_P\alpha k} = N_c N_f \cdot 2^{-L_P\alpha k}.
\]
By substituting the previously derived expressions and following the same calculation steps as in \cite{ori}, the probability \( P_{\alpha1} \) can be upper-bounded as:
\begin{equation}
    P_{\alpha1} \leq 2^{\alpha k(-2\log \alpha + 3s + 7 - L_P)}.
\end{equation}
This concludes the analysis for the case where \( \beta < \alpha < \frac{1}{2} \).

Next, we focus on the case where \(\frac{1}{2} \leq \alpha \leq R + o(\beta)\). In this scenario, the number of false pivots exceeds the number of correct pivots. Following the derivation in \cite{ori}, the number of possible configurations is given by:  
\[
\binom{sk+Rk+o(\beta)k-1}{Rk+o(\beta)k-1} \leq \binom{(s+1)k}{k} \leq 2^{(s+1)k}.
\]  
\rev{Note that once again, the segment length multiplier \(s\) appears in the derivation due to the increase in the average number of deletions per section.} Since the probability of each configuration occurring remains \(2^{-L_P\alpha k}\), the probability of having \(\alpha k\) false pivots \(P_{\alpha2}\) is bounded by:
 
\[
P_{\alpha2} \leq 2^{(s+1-L_P\alpha)k}.
\]

By considering both cases together, the probability of selecting more than \(\beta k\) false pivots is given by:  
\begin{equation}
    \label{eq:int}
    \begin{split}
        &\int_{\alpha = \beta}^{\frac{1}{2}}P_{\alpha1}d\alpha k + \int_{\alpha=\frac{1}{2}}^{R+o(\beta)}P_{\alpha2}d\alpha k\\ 
        &\leq \int_{\alpha = \beta}^{\frac{1}{2}}2^{(-2\log\alpha+3s+7-L_P)\alpha k}d\alpha k \\ 
        &\quad + \int_{\alpha=\frac{1}{2}}^{R+o(\beta)}2^{(s+1-L_P\alpha)k}d\alpha k .
    \end{split}
\end{equation}  

By setting the pivot length to satisfy \(L_P\geq 3s+8+2\log \frac{1}{\beta}\), we obtain:  
\[
(-2\log\alpha+3s+7-L_P)\alpha k\leq -\alpha k \leq -\beta k, \quad \text{for } \beta < \alpha < \frac{1}{2}
\]
\[
s+1-L_P\alpha\leq -0.5s-3, \quad \text{for } \frac{1}{2} \leq \alpha \leq R + o(\beta).
\]
Therefore, the sum of the two integrals in Equation~\eqref{eq:int} can be upper-bounded as:
\begin{equation}
    \begin{split}
        &\int_{\alpha = \beta}^{\frac{1}{2}}P_{\alpha1}d\alpha k + \int_{\alpha=\frac{1}{2}}^{R+o(\beta)}P_{\alpha2}d\alpha k \\
        &\leq \frac{k}{2} \left(2^{-\beta k} + 2^{-(3+0.5s)k}\right) = 2^{-\Omega(n)}.
    \end{split}
\end{equation}
This confirms that the probability of having more than \(\beta k\) false pivots in the selected pivots is asymptotically small when the pivot length satisfies the given constraint. Thus, Theorem~\ref{theorem:pivot length} is proven.

\end{document}